\documentclass[lettersize,journal]{IEEEtran}
\IEEEoverridecommandlockouts
\usepackage[dvipdf]{graphicx,color}
\usepackage{amssymb}
\usepackage{enumerate}
\usepackage{amsmath}
\usepackage{amsfonts}
\usepackage{balance}
\usepackage{multirow}
\usepackage{makecell}
\usepackage{color}
\usepackage{algorithm}
\usepackage{stfloats}
\usepackage{float}
\usepackage{algpseudocode}
\usepackage{color}
\usepackage{varwidth}
\usepackage{multicol}
\usepackage{subfigure}
\usepackage{xspace}
\usepackage{enumerate}
\usepackage{xcolor,cite,etoolbox}
\usepackage{bm}
\usepackage{amsthm}
\usepackage{booktabs}

\newtheorem{lemma}{Lemma}
\usepackage{multicol} 
\usepackage{array}
\newcolumntype{M}[1]{>{\raggedright\arraybackslash}m{#1}}
\newtheorem{proposition}{Proposition}
\newtheorem{corollary}{Corollary}
\newtheorem{remark}{Remark}
\usepackage[font=footnotesize,labelfont=scriptsize]{caption}
\usepackage{ragged2e}

\def\BibTeX{{\rm B\kern-.05em{\sc i\kern-.025em b}\kern-.08em
    T\kern-.1667em\lower.7ex\hbox{E}\kern-.125emX}}

\usepackage{titlesec}
\usepackage{tikz}

\newcommand{\deltaInSqcup}{
  \mathbin{
    \begin{tikzpicture}[baseline=(cup.base)]
      \node (cup) at (0,0) {$\sqcup$};
      \node at (cup.center) {\scriptsize$\delta$};
    \end{tikzpicture}
  }
}
\usepackage[margin=0.64in]{geometry}
\makeatletter 
\pretocmd\@bibitem{\color{black}\csname keycolor#1\endcsname}{}{\fail}
\newcommand\citecolor[1]{\@namedef{keycolor#1}{ \color{red}}}
\makeatother

\begin{document}

 {\title{\huge Passive RIS Is Not Silent: Revisiting Performance Limits Under Thermal Noise}}	
	\author{{Farjam Karim, \textit{Graduate Student Member}, \textit{IEEE}, Deepak Kumar, \textit{Member}, \textit{IEEE}, Prathapasinghe Dharmawansa, \\ Nurul Huda Mahmood, \textit{Member}, \textit{IEEE},  Arthur Sousa de Sena, \textit{Member}, \textit{IEEE},   Matti Latva-aho, \textit{Fellow}, \textit{IEEE}}\\
\thanks{\hrulefill}
		\thanks{
All the authors except D. Kumar and A.S de Sena are with CWC, University of Oulu, Finland. (e-mail: \{farjam.karim, Prathapasinghe.KaluwaDevage, nurulhuda.mahmood, matti.latva-aho\}@oulu.fi.).\\ 
 D. Kumar is with Department of Electronics and Communication Engineering, Motilal Nehru National Institute of Technology Allahabad, Prayagraj, 211004, India (email: dkumar@mnnit.ac.in).\\
 A.S de Sena is with Ericsson, Sweden. (e-mail: arthurssena@ieee.org).
 
  This work was supported by the Research Council of Finland through the projects 6G Flagship (grant 369116) and 6G-ConCoRSe (grant 359850), Business Finland’s 6GBridge- 6CORE Project (grant 8410/31/2022), and 6GBridge-Local 6G project (Grant No. 8002/31/2022) and  Tauno Tönningin säätiö grant.
} 

 }

	\maketitle

\begin{abstract}
Reconfigurable intelligent surfaces (RISs) have emerged as a promising solution for enabling energy-efficient and flexible spectrum usage in wireless communication, particularly in the context of sixth-generation (6G) networks. While passive RIS architectures are widely regarded as virtually noiseless due to the lack of active components, this idealized assumption can lead to misleading performance evaluations. In this paper, we revisit this assumption and demonstrate that the thermal noise generated by passive RIS elements, though often neglected, can significantly affect system performance.
We propose a tractable approximated analytical framework that incorporates RIS-induced thermal noise into the system and derive closed-form expressions for key performance metrics, such as outage probability and throughput. Simulation results validate our approximated analysis and highlight the substantial performance discrepancies that arise when RIS thermal noise is ignored.
 Our results offer valuable insights into the trade-offs between receiver and RIS noise, guiding the development of robust and efficient 6G communication systems.
	\end{abstract}
 
\vspace{-1.0em}
\section{Introduction}

Reconfigurable intelligent surfaces (RISs) have recently emerged as a promising technology for enabling energy-efficient, high-capacity wireless communication. By intelligently controlling the amplitude and phase of incident electromagnetic waves with minimal power consumption, RIS offers a flexible and low-energy solution for managing signal propagation~\cite{yLiu_com_tuto_may_21}, especially in non line-of-sight (LoS) scenarios. Various RIS architectures, such as simultaneous transmitting and reflecting (STAR)-RIS and active RIS, have been proposed to further enhance system performance. Among these, passive RIS remains the most power-efficient modeled as a noise-free due to the lack of active circuitry.
In most existing literature, the thermal noise contribution of passive RIS is completely neglected, based on its passive nature (see, e.g., Section II-A of~\cite{Zhang_tcom_23} and references therein). This simplifying assumption, while convenient, may lead to inaccurate performance evaluations and overly optimistic design choices.

In this paper, we revisit a common assumption in the RIS literature: that passive RIS elements are noise-free and their thermal noise can be neglected. While this simplification has enabled extensive analytical progress, it overlooks a fundamental physical fact, every passive element inherently generates thermal noise, as described by the Johnson–Nyquist theorem~\cite{Nyquist}. Motivated by this gap, our goal is to rigorously model this noise from first principles and assess its practical significance. Unlike in active RIS studies, where noise terms are typically included but often selected in an ad hoc manner, we present the first physically grounded thermal noise model for passive RISs, explicitly linking noise power to fundamental physical constants. We further show that incorporating this effect can alter system-level performance predictions, especially when the passive RIS and receiver are in close proximity. By doing so, this work complements existing RIS studies and provides new insights into scenarios where thermal noise is non-negligible.
\vspace{-1em}
\subsection{Notations:}
Scalars and vectors are denoted by italic and boldface letters, respectively. The notation $f_Y(y)$ and $F_Y(y)$ denote the probability density function (PDF) and cumulative distribution function (CDF) of a random variable $Y$, respectively. The Nakagami-$m$ fading model is characterized by two parameters: $m_{(\cdot)}$, which defines the fading shape, and $\Omega_{(\cdot)}$, which represents the average fading power. The notation $\mathcal{CN}(0, \sigma^2_{(\cdot)})$ denotes a circularly symmetric complex Gaussian distribution with zero mean and variance $\sigma^2_{(\cdot)}$, and $\exp(\cdot)$ represents the exponential function.
The Meijer-G function is denoted as
$ G_{m,n}^{p,q} \left( x \ \middle| \begin{matrix} 
a_1, \dots, a_n, a_{n+1}, \dots, a_p \\
b_1, \dots, b_m, b_{m+1}, \dots, b_q
\end{matrix} \right)$
as defined in~\cite[Eqn. $9.30$]{Book1}. The symbol $\mathbf{I}_N$ denotes the $N \times N$ identity matrix. The Hermitian of a vector or matrix is denoted by ${(\cdot)}^\mathrm{H}$ and $\deltaInSqcup$ indicates that positive integer value of $\delta$ is considered. Finally, $\gamma(\cdot, \cdot)$ and $\Gamma(\cdot)$ denote the lower incomplete gamma function and the complete Gamma function, respectively.
\vspace{-0.5em}
\section{System Model}\label{sec:sysmtem}
We consider a wireless communication system in which a base station (BS) communicates with a receiver via an RIS composed of $N$ passive reflecting elements. Both the BS and the receiver are equipped with single antennas. The direct LoS link between the BS and the receiver is assumed to be blocked, thereby making the RIS essential for maintaining reliable communication.
The $n$-th RIS element introduces an adjustable phase shift, denoted by $\phi_n$, where $\phi_n \in (0, 2\pi]$. The collective phase shifts introduced by the RIS can be represented using the diagonal matrix $\mathbf{\Phi} = \text{diag}(e^{j\phi_1}, e^{j\phi_2}, \dots, e^{j\phi_N})$. Furthermore, the reflection factor of the RIS  element is denoted by  $\alpha$, where $0\leq\alpha\leq 1$. All wireless channels in the system are assumed to be independent and identically distributed and follow a Nakagami-$m$ fading model.  Note that, Nakagami-$m$ fading can approximate Rayleigh, Rician, and one-sided Gaussian distributions as special cases~\cite{mallik_tcom}.
Moreover, we assume that global channel state information  is available at all nodes in the system, including the BS, RIS, and receiver\footnote{This problem has been addressed to a certain degree using PARAFAC decompositions and other algorithms~\cite{Nipuni}. However, the inclusion of thermal noise opens up a new research direction to design CSI estimation algorithms in a novel way.}. 
Next, we focus on the model used for calculating the thermal noise at the receiver and the RIS.

Thermal noise in any resistive component follows the Johnson--Nyquist relation, so the noise power over an operating bandwidth $B$ is $J_a = k \mathcal{T} B$, where $k$ is Boltzmann's constant and $\mathcal{T}$ is the physical temperature in kelvins~\cite{Nyquist, Wait_1968_TMT}. At the receiver, additional noise from active circuits is captured by the noise figure $N_F$, giving $J = k \mathcal{T} B N_F$. We model the receiver noise as a complex Gaussian variable $\omega_d \sim \mathcal{CN}(0,\sigma_d^2)$ with $\sigma_d^2 = J$, ensuring that the variance matches the physical noise power.
The same $k \mathcal{T} B$ principle applies to a passive RIS. Each reflecting element has ohmic loss and therefore generates thermal noise with power $k \mathcal{T} B$ over the operating bandwidth. Assuming the $N$ reflecting elements contribute thermal noise independently, the aggregate noise power scales linearly{\footnote{With sub-wavelength spacing, mutual coupling will introduce correlation among RIS elements, modifying the ideal linear scaling assumption of aggregate noise power with $N$.  However, the underlying $k\mathcal{T}B$ noise generation mechanism at each element remains unchanged. In this work, independence is assumed for analytical tractability. Incorporating correlated RIS noise due to mutual coupling is an interesting direction for future investigation.}} with $N$. The reflection efficiency $\alpha \in (0,1)$ denotes the fraction of both signal and noise power re-radiated toward the receiver. Therefore, the total RIS thermal noise power is $J_R = N \alpha k \mathcal{T} B$.
We represent the RIS noise vector as $\boldsymbol{\omega_r} \sim \mathcal{CN}(0, \sigma_r^2 \mathbf{I}_N)$ and set $\sigma_r^2 = k \mathcal{T} B$, so that $J_R = N\alpha \sigma_r^2$. The Gaussian assumption is standard for thermal noise because it results from the sum of many independent microscopic electron motions.
This formula makes the scaling with $B$, $N$, and $\alpha$ explicit, allowing designers to determine when RIS noise can be considered negligible and when it becomes significant. Unlike active transceivers, RIS elements do not include high power-consuming electronic components and therefore lack a conventional noise figure.
Therefore, $\alpha$ models the fraction of both the reflected signal and the internally generated thermal noise that is re-radiated toward the receiver. Consequently, lower values of $\alpha$ reduce the portion of generated noise reaching the receiver, while higher $\alpha$ increases the observed noise power. In this sense, $\alpha$ can be interpreted both as a measure of reflectivity and as an implicit indicator of signal degradation, functionally analogous to the noise figure in active systems. 
In this model, both RIS-induced noise and receiver noise are assumed to follow additive white Gaussian noise (AWGN) distributions.
 Thus, the received signal can be expressed as
\begin{align}\label{received_sig}
y_{k} = \sqrt{P_b \alpha}\left(\mathbf{g}^{\mathrm{H}}_{d}\boldsymbol{\Phi}\mathbf{g}_{b}\right)x + \sqrt{\alpha}\mathbf{g}^{\mathrm{H}}_{d}\boldsymbol{\Phi}\boldsymbol{\omega_r} + \omega_{d},
\end{align}
where $\mathbf{g}_{d} = [g_{1d}, g_{2d}, \dots, g_{Nd}]^T$, with $g_{nd}$ representing the channel gain between the $n$-th reflecting element of the RIS and the user; $\mathbf{g}_{b} = [g_{b1}, g_{b2}, \dots, g_{bN}]^T$, where $g_{bn}$ denotes the channel gain between the BS and the $n$-th RIS element; $\boldsymbol{\omega_r} \sim \mathcal{CN}(0, \sigma^2_r \mathbf{I}_N)$ is the thermal noise introduced by the RIS elements, $\omega_{d} \sim \mathcal{CN}(0, \sigma^2_d)$ is the thermal noise at the receiver, $x$ is the unit-power transmitted symbol, and $P_b$ is the transmit power of the BS. It is important to note that most previous studies ignore the impact of passive noise reflection through the RIS. Using \eqref{received_sig}, the signal-to-interference-plus-noise ratio (SINR) can be given as
\begin{align}\label{SINR}
\gamma_d = \frac{{\rho\alpha} \left|\mathbf{g}^{\mathrm{H}}_{d} \boldsymbol{\Phi} \mathbf{g}_{b}\right|^2}{\left|\mathbf{g}^{\mathrm{H}}_{d} \boldsymbol{\Phi}\right|^2 \lambda + 1},
\end{align}
where $\rho =\frac{P_b}{\sigma^2_d}$ is the transmit signal-to-noise ratio (SNR) and  $\lambda = \frac{\alpha\sigma^2_r}{\sigma^2_d}$. Note that \eqref{SINR} differs from prior passive RIS formulations because it includes the random channel multiplied by RIS-generated thermal noise. 
\section{Performance Evaluation}\label{sec_3}
In this section, we analyze the performance of the passive RIS-assisted wireless network by deriving an analytical expression for the outage probability, defined as the probability that the received SINR falls below a specified threshold, leading to a link failure. To make the analysis tractable, we introduce bounds on the SINR. Specifically, we define a lower bound (LB) and an upper bound (UB) for the SINR in \eqref{SINR} as
\begin{align}\label{lower_bound}
    \gamma^{LB} =\frac{1}{2}\text{min}\left(\frac{{\rho\alpha} \left|\mathbf{g}^{\mathrm{H}}_{d} \boldsymbol{\Phi} \mathbf{g}_{b}\right|^2}{\left|\mathbf{g}^{\mathrm{H}}_{d} \boldsymbol{\Phi}\right|^2 \lambda}, {\rho\alpha} \left|\mathbf{g}^{\mathrm{H}}_{d} \boldsymbol{\Phi} \mathbf{g}_{b}\right|^2 \right),
\end{align}
and
\begin{align}\label{upper_bound}
    \gamma^{UB} =\text{min}\left(\frac{{\rho\alpha} \left|\mathbf{g}^{\mathrm{H}}_{d} \boldsymbol{\Phi} \mathbf{g}_{b}\right|^2}{\left|\mathbf{g}^{\mathrm{H}}_{d} \boldsymbol{\Phi}\right|^2 \lambda}, {\rho\alpha} \left|\mathbf{g}^{\mathrm{H}}_{d} \boldsymbol{\Phi} \mathbf{g}_{b}\right|^2 \right).
\end{align}
\begin{figure*}[t]
\begin{align}\nonumber
     P_{Od}&= \mathrm{Pr}\big[ \gamma^{UB}< \Upsilon_{{th}}\big] = \mathrm{Pr}\bigg[ \text{min}\left(\frac{{\rho\alpha} \left|\mathbf{g}^{\mathrm{H}}_{d} \boldsymbol{\Phi} \mathbf{g}_{b}\right|^2}{\left|\mathbf{g}^{\mathrm{H}}_{d} \boldsymbol{\Phi}\right|^2 \lambda}, {\rho\alpha} \left|\mathbf{g}^{\mathrm{H}}_{d} \boldsymbol{\Phi} \mathbf{g}_{b}\right|^2 \right) < \Upsilon_{{th}}\bigg]
\end{align}
\begin{align}\label{approx_2}
  &=  1 - \left[ 1 - \underbrace{\mathrm{Pr} \left( \frac{{\rho\alpha} \left|\mathbf{g}^{\mathrm{H}}_{d} \boldsymbol{\Phi} \mathbf{g}_{b}\right|^2}{\left|\mathbf{g}^{\mathrm{H}}_{d} \boldsymbol{\Phi}\right|^2 \lambda} \right) < \Upsilon_{{th}}}_{\xi_1} \right] \left[ 1 - \underbrace{\mathrm{Pr} \left( {\rho\alpha} \left|\mathbf{g}^{\mathrm{H}}_{d} \boldsymbol{\Phi} \mathbf{g}_{b}\right|^2 \right) < \Upsilon_{{th}}}_{\xi_2} \right].
\end{align}\hrulefill \vspace{-1em}
\end{figure*}
The UB and LB in \eqref{lower_bound} and \eqref{upper_bound}, respectively, can be derived by following the steps in~\cite[Section~9]{UB_LB}. The details are omitted here due to space limitations.
For the considered network, the LB of the outage probability can be determined by comparing the UB of the SINR with the threshold $\Upsilon_{th}$, as given in \eqref{approx_2} at the top of the next page. Here, $\Upsilon_{th} = 2^{\mathfrak{R}/B} - 1$, where $\mathfrak{R}$ denotes the target data rate and $B$ is the operating bandwidth. The expression for the outage probability at the receiver is given below, where the terms $\xi_1$ and $\xi_2$ are derived in Lemma 1 and Lemma 2, respectively.
\begin{proposition}
    The approximated outage probability at the receiving device can be evaluated as
    \begin{align}\label{out_final}
       P_{Od} = 1- \left[\left(1-\xi_1\right) \left(1-\xi_2\right)\right],
    \end{align}
   where $\xi_1$ and $\xi_2$ are given in \eqref{lemma_1} and \eqref{lemma2}, respectively. Note that the same expression can be used to find the outage probability UB by replacing $\Upsilon_{th}$ with $2\Upsilon_{th}$. 
\end{proposition}
\begin{remark}
    If thermal noise is neglected, then \eqref{out_final} reduces to $\xi_2$ term only, which corresponds to the commonly used noiseless RIS model in prior works~\cite{Zhang_tcom_23}.
\end{remark}
\begin{lemma}
    The term $\xi_1$ can be expressed as 
    \begin{align}\label{lemma_1}
        &\xi_1= 1- \sum\limits^{\deltaInSqcup\!- 1}_{p=0} \frac{1}{p!\sqrt{\pi}} \left(\frac{m_{nd}}{\Omega_{nd}}\right)^{-\frac{p}{2}}\left(\frac{1}{\zeta}\sqrt{\frac{\Upsilon_{\text{th}}\lambda}{\psi}}\right)^p\nonumber\\
        &\times \frac{1}{\Gamma(m_{nd}N)}G_{1,2}^{2,1} \left( \frac{\lambda\Omega_{nd}\Upsilon_{\text{th}}}{4\psi m_{nd}\zeta^2} \Bigg| \begin{array}{c} 1-m_{nd}N-\frac{p}{2} \\0, \frac{1}{2} \end{array} \right),
    \end{align}
    where $\mu_\mathfrak{D} = \frac{\Gamma\left(m_{bn}+0.5\right)\Gamma\left(m_{nd}+0.5\right)}{\Gamma\left(m_{bn}\right)\Gamma\left(m_{nd}\right)}\left(\frac{\Omega_{bn}\Omega_{nd}}{m_{bn}m_{nd}}\right)^{\frac{1}{2}}N$, $\sigma^2_{\mathfrak{D}} = N{\Omega_{bn}\Omega_{nd}}\Bigl[1-\frac{1}{m_{bn}m_{nd}} {\left(\frac{\Gamma\left(m_{bn}+0.5\right)\Gamma\left(m_{nd}+0.5\right)}{\Gamma\left(m_{bn}\right)\Gamma\left(m_{nd}\right)}\right)^2}\Bigl]$, $\psi=\rho\alpha$, $\delta= \frac{\mu^2_{\mathfrak{D}}}{\sigma^2_{\mathfrak{D}}}$, $\zeta=\frac{\sigma^2_{\mathfrak{D}}}{\mu_{\mathfrak{D}}}$, $m_{bn}$, $m_{nd}$, $\Omega_{bn}$, and $\Omega_{nd}$ denotes the shape parameter and average power of the fading from the BS to the $n$-th RIS element and $n$-th RIS element to the receiver, respectively.
\end{lemma}
\begin{remark}
    From \eqref{lemma_1}, $P_{Od}$ depends on $\lambda$, which scales with $N$, $B$, and $\alpha$. Thus, larger RISs may boost signal gain but also increase thermal noise, reducing net performance benefits.
\end{remark}
\begin{proof}
    Please refer to Appendix A.
\end{proof}

\begin{lemma}
    The term $\xi_2$ can be expressed as 
    \begin{align}\label{lemma2}
      \xi_2=  \frac{1}{\Gamma(\delta)} \, G^{1,1}_{1,2} \left( \frac{1}{\zeta} \sqrt{\frac{\Upsilon_{\text{th}}}{\psi}} \;\middle|\; \begin{array}{c} 1 \\ \delta, 0
\end{array} \right).
    \end{align}
\end{lemma}
\begin{proof}
     The proof follows the same steps as in Appendix A up to equation~\eqref{CDF_proof}, after which the lower incomplete gamma function is expressed using the Meijer G-function.
\end{proof}
\begin{remark}
      The proposed analysis focuses on passive RIS, where each element reflects the incident signal with $\alpha\leq 1$ and generates only thermal noise due to resistive losses. If extended to an active RIS with an uniform amplification of  $\alpha> 1$, each element would amplify both the incident signal and the RIS-generated thermal noise, and additionally introduce amplifier noise from the RIS active circuitry. In such a case, the received signal model must explicitly include this additional noise term, resulting in modified value for $\lambda$. Therefore, although the structural form of \eqref{SINR} remains similar, the aggregate noise statistics differ due to the presence of amplifier noise.	
\end{remark}
Since \eqref{out_final} involves special mathematical function (e.g., the Meijer-G function), we present an asymptotic analysis of the outage probability to obtain deeper insights into the system.
\begin{proposition}
    The asymptotic outage probability at the receiving device can be approximated as
    \begin{align}\label{out_final_asymptoic}
       P_{Od}^{\infty} = 1- \left[\left(1-\xi_1^{\infty}\right) \left(1-\xi_2^{\infty}\right)\right],
    \end{align}
   where $\xi_1^{\infty}$ and $\xi_2^{\infty}$ are given in \eqref{Asymp_1} and \eqref{Asymp_2}, respectively. 
   \end{proposition}
   \begin{remark}
       The asymptotic outage behaves as $P_{Od}^{\infty} \sim \rho^{-\delta/2}$ at high SNR, giving a diversity order of $\delta/2$.
   \end{remark}
   \begin{lemma}
       The term $\xi_1^{\infty}$ can be expressed as 
       \begin{align}\label{Asymp_1}
           \xi_1^{\infty}=\left(\frac{m_{nd}}{\Omega_{nd}}\right)^{-\frac{\delta}{2}}\frac{\Gamma\left(m_{nd}N+\frac{\delta}{2}\right)}{\delta\Gamma(\delta)\Gamma\left(m_{nd}N\right)}\left(\frac{1}{\zeta}\sqrt{\frac{\Upsilon_{\text{th}}\lambda}{\psi}}\right)^\delta.
       \end{align}
   \end{lemma}
   \begin{lemma}
       The term $\xi_2^{\infty}$ can be expressed as 
       \begin{align}\label{Asymp_2}
           \xi_2^{\infty}=\frac{1}{\delta\;\Gamma(\delta)}\left(\frac{1}{\zeta}\sqrt{\frac{\Upsilon_{\text{th}}\lambda}{\psi}}\right)^\delta.
       \end{align}
   \end{lemma}
      \begin{proof}
Both Lemma 3 and Lemma 4 follow the approach in Appendix A up to~\eqref{step1}.  
We then approximate the lower incomplete gamma function as 
$\gamma(c,z) \underset{z \to 0}{\approx} \tfrac{z^c}{c}$~\cite{karim_wcl}.  
For Lemma 4, this approximation alone yields the asymptotic expression $\xi_2^{\infty}$ in \eqref{Asymp_2}.  
For Lemma 3, we proceed further by solving the resulting integral using~\cite[Eq.~3.381.4]{Book1}, which gives the full expression for $\xi_1^\infty$.
   \end{proof}
\begin{corollary}
    The throughput at the receiving device can be given as
     $ \mathfrak{T} = \left(1-P_{Od}\right)\mathfrak{R}$,
    where $\mathfrak{R}$ is the target data rate.
\end{corollary}
\section{Numerical Results and Discussions}
This section presents Monte Carlo simulation results to validate the derived closed-form expressions.
The wireless link between the BS and receiver is characterized using a distance-based path-loss model, given by $\varphi/D^{\tau_{i}}$, where  $\varphi = 1$~meter is the reference path-loss parameter, $D_i$ is the distance and $\tau_{i}\;\forall i \in \{{bn}, {nd}\} $ is the path-loss exponent between the BS to RIS and RIS to receiver.  
 \begin{table}[t]	\renewcommand{\arraystretch}{1.0}
		\centering
		\caption{ Simulation Parameters.}
		\label{t3}
			\resizebox{\columnwidth}{!}{\begin{tabular}{|l|l|l|l|l|l|}
			\hline
			Parameter         & Value         & Parameter & Value  & Parameter & Value  \\ \hline
			$m_{bn}  $   &     $2 $   	&  $\mathfrak{R}$   &     $15$~Mbps & $\sigma^2_d$   &     $-128 $ dB \\ \hline
			
			$D_{bn}$    &    $ 100 $~m &	  $\tau_{bn}$ &     $3.2$ &   $\tau_{nd}$ &     $2$  	 \\ \hline	

           $m_{nd}  $   &     $2 $ &	  $B$ &     $20$~MHz &   $\alpha $ &     $0.9 $  	 \\ \hline	

		\end{tabular}}\vspace{-1em}
	\end{table}
    Unless otherwise stated, the simulation parameters are summarized in Table~\ref{t3}. Additionally, Table~\ref{thermal_table} presents the thermal noise values contributed by the RIS for different numbers of RIS elements corresponding to their amplification factors. These simulation results offer key insights into how thermal noise introduced by the RIS affects overall system performance. The AWGN at the receiver is computed using the formula provided in  Section~\ref{sec:sysmtem}, and $N_F$ is set to $3$~dB\footnote{Modern low-noise amplifier  designs, such as the one presented in~\cite{lna_5g}, report even lower NF values (e.g., 1.3–1.4~dB) for $ 5G$ new-radio applications, supporting the feasibility of assuming $3$~dB for our case.}.
    
   \begin{table}[h!]
\centering
\caption{RIS Thermal Noise in dB at 20 MHz Bandwidth.}
\begin{tabular}{|c|c|c|c|}
\hline
 $\alpha$ & $N=5$ & $N=10$ & $N=20$\\ \hline
0.1 & -133.9772  & -130.9669  & -127.9566   \\ \hline
0.2 &-130.9669   &  -127.9566  & -124.9463  \\ \hline
0.3 &-129.2060  &  -126.1957  &  -123.1854    \\\hline
0.4 &-127.9566   & -124.9463   & -121.9360     \\ \hline
0.5 &-126.9875 &  -123.9772   &-120.9669      \\ \hline
0.6 &-126.1957 & -123.1854    &-120.1751      \\ \hline
0.7 & -125.5262   & -122.5159    &-119.5056       \\ \hline
0.8 &-124.9463 & -121.9360    &-118.9257       \\ \hline
0.9 &-124.4348 & -121.4245    & -118.4142      \\ \hline
1.0 & -123.9772 & -120.9669   &-117.9566     \\ \hline
\end{tabular}
\label{thermal_table}
\end{table}
\begin{figure*}[t]
    \centering
    \begin{minipage}[b]{0.32\textwidth}
        \centering
        \includegraphics[width=\textwidth]{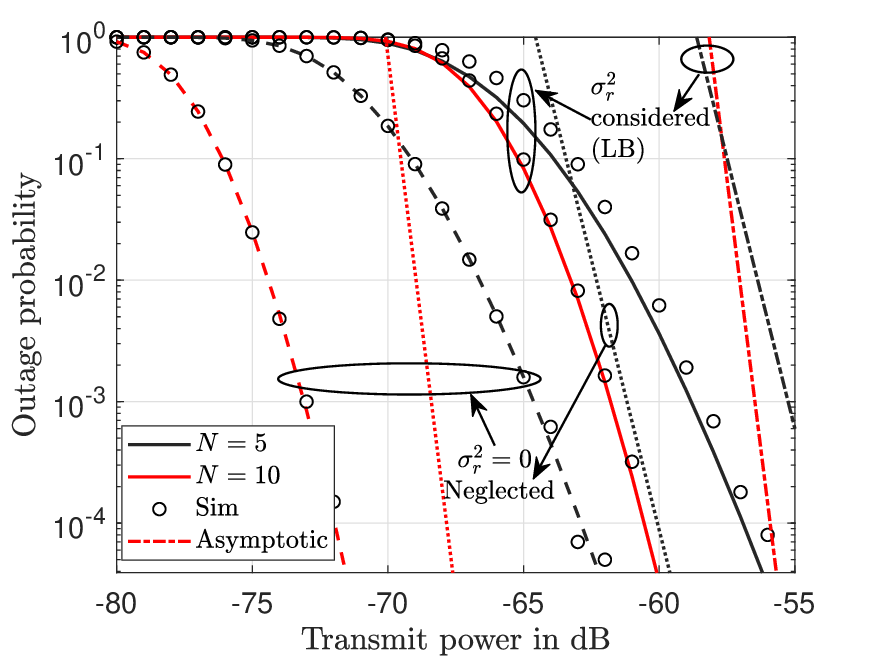}
        \caption{$P_{Od}$ versus $P_b$ (Approximation Holding).}
        \label{approx_holding}
    \end{minipage}
     \begin{minipage}[b]{0.32\textwidth}
        \centering
        \includegraphics[width=\textwidth]{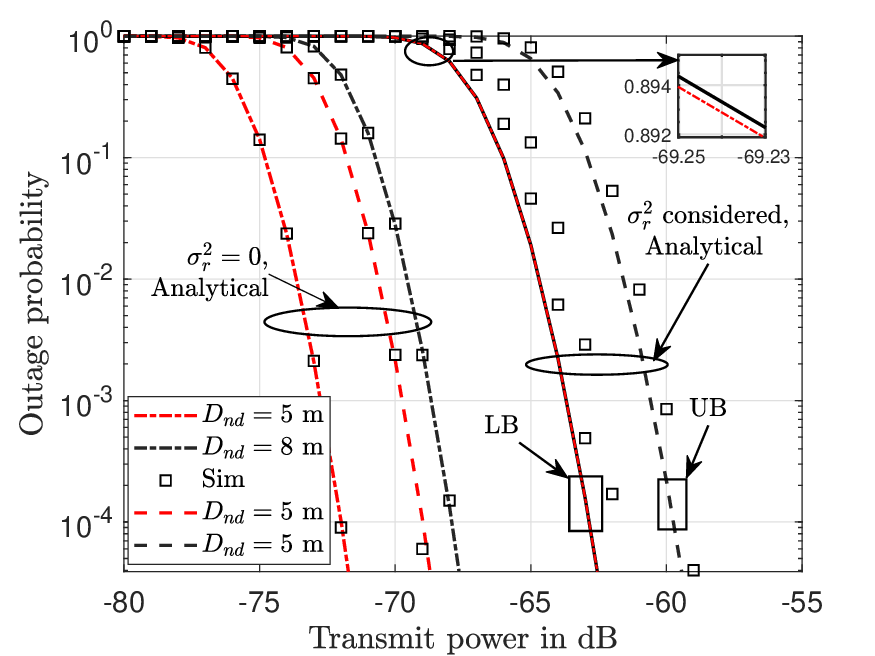}
        \caption{$P_{Od}$ versus $P_b$ (Approximation Deviates).}
        \label{not_holding}
    \end{minipage}
     \begin{minipage}[b]{0.32\textwidth}
        \centering
        \includegraphics[width=\textwidth]{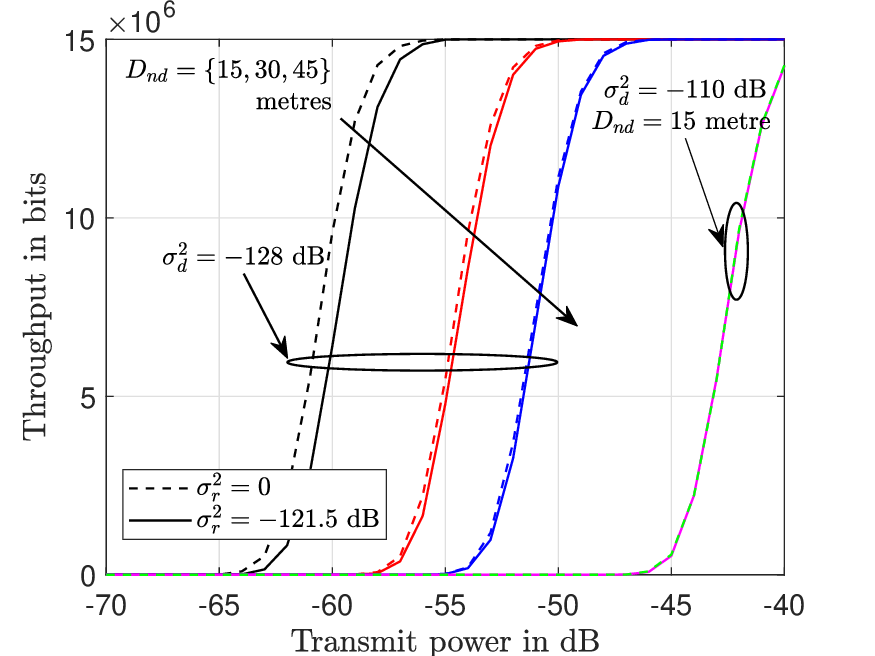}
        \caption{Effect of distance on Throughput.}
        \label{thro}
    \end{minipage}
    
\end{figure*}

 Fig.~\ref{approx_holding} illustrates the  outage probability and the asymptotic behavior of the receiver as a function of the transmit power ($P_b$) at a distance of 2~meters from the RIS. The plot considers two scenarios: one in which thermal noise from the RIS is included, and another where it is neglected. As expected, increasing $P_b$ results in improved (i.e., lower) outage probability. A similar trend is observed when the number of reflecting elements, $N$, increases from $5$ to $10$, highlighting the performance enhancement due to larger RIS dimensions.
Beyond these expected observations, the figure provides critical insights into the impact of RIS thermal noise. First, neglecting thermal noise leads to a significant overestimation of system performance, and this discrepancy becomes more pronounced as $N$ increases. For instance, when $N = 5$, an outage probability of $1 \times 10^{-3}$ is achieved at $P_b \approx -59$~dB when RIS thermal noise ($\sigma_r^2$) is accounted for, whereas the same performance is achieved at $P_b \approx -64.5$~dB when $\sigma_r^2 = 0$. There is an overestimation of about $5$ dB for $N = 5$, which further widens to about  $12$ dB for $N = 10$. This highlights the increasing overestimation of system performance when RIS thermal noise is neglected with more reflecting elements. Second, consider the case when $\sigma_r^2 = 0$ and $N$ is increased from $5$ to $10$: the outage probability of $1 \times 10^{-3}$ is achieved at $P_b \approx -64.6$~dB and $P_b \approx -73$~dB, respectively. This indicates that increasing $N$ by $5$ improves energy efficiency by approximately $8.4$~dB. However, when $\sigma_r^2$ is considered, the same performance requires $P_b \approx -58.8$~dB and $P_b \approx -61.8$~dB, respectively, resulting in only a $3$~dB improvement. 
This discrepancy raises an important question: \textit{Are the energy efficiency gains of RIS as great as is widely emphasized in the literature?} Or should performance analyses be revisited with proper consideration of RIS thermal noise to realistically evaluate the benefits of RIS deployment in future wireless networks?  It should be highlighted here that we have ignored other non-idealities such as hardware impairments and imperfect CSI, which will further reduce the performance gain with RIS.
The close agreement between the simulation and analytical curves verifies the accuracy of the approximated lower-bound outage probability in~\eqref{out_final}. Similarly, the asymptotic analysis aligns with both simulation and analytical results as $P_b \to \infty$, confirming the validity of the expression in~\eqref{out_final_asymptoic}. Importantly, the asymptotic results reveal the critical impact of RIS thermal noise, which is less evident in the standard outage curves. For instance, there exists a region where the case with $N=5$ outperforms $N=10$, a trend also faintly visible around $P_b \approx -69$~dB in the normal outage curve but more clearly highlighted in the asymptotic behavior. This phenomenon arises because increasing $N$ amplifies the thermal noise contribution from the RIS, and its effect becomes particularly pronounced in the low transmit power regime.

Fig.~\ref{not_holding} presents the outage probability as a function of $P_b$ for receiver distances of $5$ and $8$~meters from the RIS, considering both $\sigma^2_r = 0$ and $\sigma^2_r = -118.4$~dB, with $N = 20$ reflecting elements. At $5$m, the analytical approximation closely matches the simulation results, confirming the validity of our closed-form expression under moderate RIS–receiver distances. However, at $8$m, a noticeable deviation emerges, between the analytical and simulated curves. This behavior highlights a limitation of the derived approximation in scenarios with larger RIS–receiver distances. The root cause lies in the parameter $\lambda$, which decreases as the distance grows. A very small $\lambda$ affects the  expression in~\eqref{lemma_1}, propagating into~\eqref{out_final} and reducing the tightness of the bounds. However, we would still like to emphasize that the simulation results for $\sigma^2_r = -118.4$~dB clearly demonstrate a performance gap between the $5$ and $8$~meter cases. This supports the fundamental premise of our work that thermal noise generated by passive RIS elements has a non-negligible impact and should not be ignored in performance analysis. Moreover, the UB curves are included for reference and, as expected, consistently lie above both the simulation and LB analytical results, offering a conservative yet analytically useful performance estimate particularly at the $5$~meter distance, where the bounds are relatively tight. Therefore, from Fig.~\ref{not_holding}, we conclude that, under the considered parameter set, the proposed analytical framework remains in close agreement with the simulation results for distances up to $5$ m.

 Fig.~\ref{thro} illustrates the simulated throughput measured in bits at the receiver as a function of transmit power ($P_b$), considering different RIS to receiver distances and receiver noise power levels ($\sigma^2_d$), with $N=10$ reflecting elements. This figure serves two primary purposes. First, it aims to identify the distance range within which the thermal noise generated by the RIS remains non-negligible. Second, it evaluates the influence of $\sigma^2_d$ on the effective system parameter $\lambda$, which governs the impact of RIS thermal noise on performance.
The results show that for a low receiver noise level ($\sigma^2_d = -128$~dB), the RIS thermal noise significantly affects throughput up to a distance of approximately $30$~meters. However, as the distance increases to $45$~meters, the impact of RIS thermal noise becomes negligible for the same $\sigma^2_d$. In contrast, for a relatively high value of  $\sigma^2_d = -110$~dB, even at a relatively short distance of $15$~meters, the influence of RIS thermal noise on throughput is is negligible, i.e., the RIS reflected thermal noise is drowned by the receiver thermal noise. 
This behavior is explained by the formulation of $\lambda$, which inherently incorporates both $\sigma^2_d$ and $\sigma^2_r$. As $\sigma^2_d$ increases, the relative contribution of $\sigma^2_r$ diminishes, thereby reducing the observed effect of RIS thermal noise.


\section{Conclusion and Future research directions}
In this work, we challenged the prevailing assumption that passive RIS operate without introducing meaningful thermal noise into wireless systems.
While RIS is often proposed as a beneficial solution for enhancing energy efficiency, our findings suggest that its utility is \textit{highly context-dependent}. Practical questions now emerge:
\begin{itemize}
    \item Under what exact conditions does RIS deployment yield a net performance gain?
    \item Should RIS be avoided in ultra-low-noise applications unless thermal effects are properly mitigated?
    \item Can RIS element designs be optimized to minimize thermal contributions without sacrificing reflection efficiency?
\end{itemize}
Additionally, we would like to emphasize the need for new mathematical frameworks that move beyond traditional bounding techniques. Our results indicate that the commonly used approximated expressions for RIS-aided systems begin to break down at larger RIS–receiver distances. This highlights the importance of developing more accurate, non-asymptotic analytical tools for performance analysis in future RIS-enabled networks.
Ultimately, our findings emphasize that RIS cannot be treated as a noise-free black box in performance analysis.

{\appendices
{\section{}
Let $X\triangleq \left(\sum\limits^{N}_{n=1}\left|g_{bn}\right|\left|g_{nd}\right|\right)^2$ assuming ideal constructive phase alignment across RIS elements. The mean and the variance of $\left(\sum\limits^{N}_{n=1}\left|g_{bn}\right|\left|g_{nd}\right|\right)$ can be expressed as~\cite{dEEPAK_TWC_25}
\begin{align}\label{mean}
    \mu_\mathfrak{D} = \frac{\Gamma\left(m_{bn}+0.5\right)\Gamma\left(m_{nd}+0.5\right)}{\Gamma\left(m_{bn}\right)\Gamma\left(m_{nd}\right)}\left(\frac{\Omega_{bn}\Omega_{nd}}{m_{bn}m_{nd}}\right)^{\frac{1}{2}}N.
\end{align}
\begin{align}\label{variance}
    &\sigma^2_{\mathfrak{D}} = N{\Omega_{bn}\Omega_{nd}}\nonumber\\
    &\times\Bigl[1-\frac{1}{m_{bn}m_{nd}} {\left(\frac{\Gamma\left(m_{bn}+0.5\right)\Gamma\left(m_{nd}+0.5\right)}{\Gamma\left(m_{bn}\right)\Gamma\left(m_{nd}\right)}\right)^2}\Bigl].
\end{align}
Next, we can approximate $X$ as Gamma distribution whose CDF can be expressed as~\cite{dEEPAK_TWC_25}
\begin{align}\label{CDF_proof}
    F_{X}(x) =\frac{1}{\Gamma(\delta)}\gamma\left(\delta, \frac{\sqrt{x}}{\zeta}\right),
\end{align}
where $\delta$ and $\zeta$ are the shape and scale parameter, respectively. Next, let $Y\triangleq \left(\sum\limits^{N}_{n=1}\left|g_{nd}\right|^2\right)$. The PDF of $Y$ can be expressed as 
\begin{align}
    f_{Y}(y)= \left(\frac{m_{nd}}{\Omega_{nd}}\right)^{m_{nd}N}\frac{y^{m_{nd}N-1}}{\Gamma(m_{nd}N)}\exp{\left(-\frac{m_{nd}}{\Omega_{nd}}y\right)}.
\end{align}
Using \eqref{approx_2}, the term $\xi_1$ can be can be written as 
\begin{align}\label{step1}
    \xi_1&=\int^{\infty}_{0}f_{Y}(y)F_{X}\left(\frac{\lambda\Upsilon_{\text{th}}y}{\psi}\right)dy.
\end{align}
By substituting the PDF and CDF in \eqref{step1} and representing $\exp(-\sqrt{by})=\frac{1}{\sqrt{\pi}}G_{0, 2}^{2,0} \left( \frac{by}{4} \Bigg| \begin{array}{c}  \\0, \frac{1}{2} \end{array} \right)$, we obtain
\begin{align}\label{step2}
   & \xi_1 \!=\! 1\!-\! \sum\limits^{\deltaInSqcup \!- 1}_{p=0} \frac{1}{p!\sqrt{\pi}} \left(\frac{m_{nd}}{\Omega_{nd}}\right)^{-m_{nd}N}\!\left(\frac{1}{\zeta}\sqrt{\frac{\Upsilon_{\text{th}}\lambda}{\psi}}\right)^p \frac{1}{\Gamma(m_{nd}N)}\nonumber\\
        &\times\int^{\infty}_{0}\!\!\! y^{{m_{nd}N}+\frac{p}{2}-1}\exp \left(\frac{-m_{nd}}{\Omega_{nd}}y\right) G_{0, 2}^{2,0} \left( \frac{\lambda\Upsilon_{\text{th}}y}{4\psi\zeta^2} \Bigg| \begin{array}{c}0 \\0, \frac{1}{2} \end{array} \right).
\end{align}
Finally  by using~\cite[Eq. $07.34.21.0088.01$]{wolfram}, the resulting integral in \eqref{step2} can be solved. Thus, we obtain \eqref{lemma_1}.}  
}

\section*{Acknowledgment}
The authors would like to thank Dr. Joonas Kokkoniemi for his insightful discussions and valuable feedback.
\bibliographystyle{IEEEtran_renamed}
\bibliography{ref}
\end{document}